\documentclass[12pt,reqno]{amsart}

\usepackage[utf8]{inputenc}
\usepackage{amssymb}
\usepackage{fullpage}
\usepackage{complexity}
\usepackage{mathtools}
\usepackage{enumitem}
\usepackage[colorlinks,linkcolor=blue,citecolor=magenta]{hyperref}
\usepackage{tikz}
\usepackage{pstricks}

\newtheorem{theorem}{Theorem}[section]
\newtheorem{proposition}[theorem]{Proposition}
\newtheorem{lemma}[theorem]{Lemma}

\theoremstyle{definition}

\def\ds{\displaystyle}
\renewcommand{\geq}{\geqslant}
\renewcommand{\leq}{\leqslant}
\renewcommand{\preceq}{\preccurlyeq}

\def\airports{\mathcal{P}}

\def\calD{\mathcal{D}}
\def\calV{\mathcal{V}}
\def\calA{\mathcal{A}}
\def\boldx{\boldsymbol{x}}
\def\boldy{\boldsymbol{y}}
\def\bolds{\boldsymbol{s}}

\def\bbZ{\mathbb{Z}}
\def\un{\boldsymbol{1}}

\title{Aircraft routing: periodicity and complexity}

\author{Frédéric Meunier}
\address{CERMICS, ENPC, Institut Polytechnique de Paris, Marne-la-Vallée, France}
\email{frederic.meunier@enpc.fr}

\author{Axel Parmentier}
\address{CERMICS, ENPC, Institut Polytechnique de Paris, Marne-la-Vallée, France}
\email{axel.parmentier@enpc.fr}

\author{Nour ElHouda Tellache}
\address{N.E.H. Tellache, Decision Support \& Operations Research Group, Department of Informatics, University of Fribourg, Fribourg, Switzerland}
\email{nourelhouda.tellache@unifr.ch}


\begin{document}
	\begin{abstract}
   The aircraft routing problem is one of the most studied problems of operations research applied to aircraft management. It involves assigning flights to aircraft while ensuring regular visits to maintenance bases. This paper examines two aspects of the problem. 
   
   First, we explore the relationship between periodic instances, where flights are the same every day, and periodic solutions. The literature has implicitly assumed—without discussion—that periodic instances necessitate periodic solutions, and even periodic solutions in a stronger form, where every two airplanes perform either the exact same cyclic sequence of flights, or completely disjoint cyclic sequences. However, enforcing such periodicity may eliminate feasible solutions. We prove that, when regular maintenance is required at most every four days, there always exist periodic solutions of this form.  

    Second, we consider the computational hardness of the problem.  Even if many papers in this area refer to the \NP-hardness of the aircraft routing problem, such a result is only available in the literature for periodic instances. We establish its \NP-hardness for a non-periodic version. Polynomiality of a special but natural case is also proven.
\end{abstract}

\maketitle

\section{Introduction}

\subsection{Aircraft routing} 

Aircraft management presents many interesting problems for operations research, some of which are especially appealing for theoretical investigation. One prominent example is the aircraft routing problem. In this problem, we are given a fleet of airplanes and a collection of flight legs to be assigned to them. The challenge lies in the requirement that every airplane must undergo regular maintenance, which can only be performed at specific airports, known as bases. The constraint of regular maintenance translates then into the constraint that every airplane must visit regularly such bases.

This paper explores two versions of the aircraft routing problem, both of which are practically relevant and have been studied in the literature. In the first version, the instance is periodic, meaning the flight legs are identical each day (often referred to as the ``daily fleet assignment problem'' in the literature~\cite{Talluri2015}). Previous studies have only been looking at periodic solutions, where each airplane repeatedly follows the same sequence of flight legs~\cite{gopalan1998aircraft,Talluri2015}. Actually, they even focused on ``Maintenance Circuit Decompositions,'' which is a more demanding type of periodicity, where every two airplanes perform either the exact same cyclic sequence of flight legs, or completely disjoint cyclic sequences. In the present paper, we refer to such solutions as ``absolutely periodic.''  

There is no inherent reason to assume that periodic or even absolutely periodic solutions always exist. Some instances might appear infeasible with no periodic solutions available, even though non-periodic solutions exist. This leads to a natural question: does the existence of a non-periodic solution imply that a periodic or even absolute periodic one must also exist? We provide a complete answer for cases where airplanes require maintenance at most every $\gamma$ days, with $\gamma \leq 4$. Such a constraint is commonly encountered in practice: the value of $\gamma$ depends on the aircraft type and the airline, but in most cases, it is three or four, making our result particularly relevant~\cite{liang2015robust, Talluri2015}.

In the second version---which we call ``finite-horizon''---there is a finite collection of precisely scheduled flight legs
, each to be performed only once. The question here is whether a feasible solution exists given the current fleet. This is arguably the most standard version of the aircraft routing problem. We prove that this version is \NP-complete when $\gamma \geq 4$, where $\gamma$ still stands for the maximum number of days without maintenance. This result is neither just another hardness result nor an anecdotal one: many operations research papers devoted to solving practically the aircraft routing problem justify their approaches by referencing its \NP-hardness; see, e.g.,~\cite{liang2009aircraft,yurek2024combinatorial}. However, as far as we know, the only existing hardness result in the literature is from Talluri~\cite{Talluri2015}, established for ``Maintenance Euler Tours,'' which is the ``periodic'' version introduced above with an added condition on another type of maintenance; see Section~\ref{subsec:period} for the exact definition.

\subsection{Periodicity}\label{subsec:period}

We introduce now more precisely the periodic version of the aircraft routing problem considered in this paper, which we call the \emph{periodic aircraft routing problem}. We start by defining it mathematically and we will then relate this definition with the concrete problem met in practice.

Assume given an Eulerian directed graph $D=(V,A)$, a subset $B\subseteq V$, and an integer $\gamma \geq 1$. Denote by $n$ the cardinality of $V$ and by $m$ that of $A$. The problem aims at finding a collection of $m$ semi-infinite walks $W_1,\ldots,W_m$ such that
\begin{itemize}
    \item each walk visits $B$ infinitely many times, with at most $\gamma$ arcs between two visits.
    \item the $j$th arcs of the walks are pairwise distinct and form together $A$, for all $j=1,2,\ldots$.
\end{itemize}
(A {\em semi-infinite} walk is a walk with infinitely many vertices, that has a first vertex but no last vertex.) Note that the assumption that $D$ is Eulerian is not restrictive: if there exists a solution, then necessarily the in- and outdegrees of each vertex are equal, and if the graph is not connected, we can work on  each component independently.

We relate now this mathematical problem with the real-life one. Each vertex is an airport. The vertices in $B$ are the bases. The number $\gamma$ is the maximal number of flying days with no maintenance. Each arc of the graph represents a sequence of flight legs that can be operated by an airplane over one day. Such sequences are usually called lines of flying (LOFs); see Gopalan and Talluri~\cite{gopalan1998aircraft} and Talluri~\cite{Talluri2015}. In a preliminary step, the flight legs are partitioned into LOFs, and the aircraft routing problem deals with LOFs rather than individual flight legs, resulting in a more tractable version of the problem. This version is called periodic because the instance is periodic in the following sense: these LOFs remain the same each day, and each must be performed daily by an airplane.

We turn now to the possible periodicity of the solutions, and we introduce the following definitions. A solution is \emph{periodic} if each $W_i$ is periodic. It is \emph{strongly periodic} if each $W_i$ is periodic and does not visit twice a same arc over the period. In other words, each $W_i$ consists in following a same closed directed trail repeatedly. It is \emph{absolutely periodic} if it is strongly periodic and every two $W_i$'s are either disjoint or identical. Only the notions of periodic and absolutely periodic solutions are considered in this paper. We have introduced the notion of strongly periodic solutions for the sake of consistency. See Figure~\ref{fig:periodicsol} for an illustration of these notions. Note, as also illustrated in the figure, the period of the $W_i$ might differ.

\begin{figure}[h]
\begin{center}
			\psscalebox{1.0 1.0} 
			{
			\includegraphics{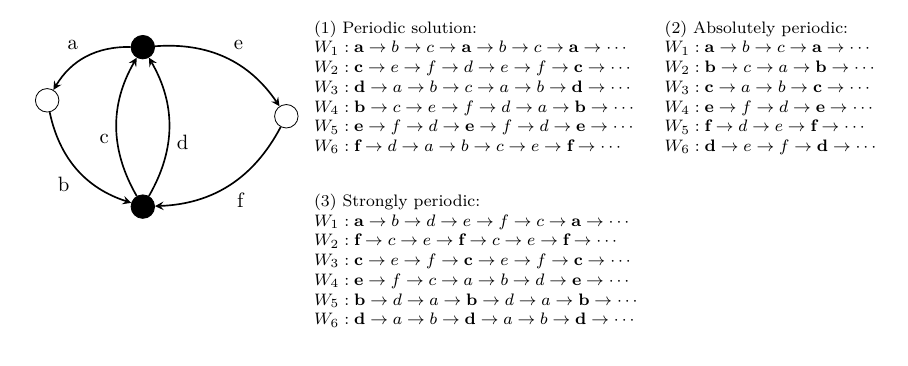}
			}
		\end{center}

\caption{On the left: an Eulerian directed graph $D$ with the black vertices representing the vertices of $B$. On the right: examples of periodic, absolutely periodic, and strongly periodic solutions. Bold vertices denote period start. (1) is not strongly periodic because $e$ is visited twice in every period of $W_2$. (3) is not absolutely periodic as $W_1$ and $W_2$ are neither disjoint nor identical.}
\label{fig:periodicsol}
\end{figure}

Periodic solutions exist as soon as the problem is feasible; see Section~\ref{sec:periodic}. All the question is about the existence of absolutely periodic solutions.

An absolutely periodic solution is one in which each LOF is always followed by the same LOF, regardless of the airplane or day. Such a solution consists of a collection of closed directed trails whose arc sets partition $A$, with each trail visiting $B$ at most every  $\gamma$ arcs. In the terminology of Talluri~\cite{Talluri2015}, absolutely periodic solutions correspond exactly to \emph{Maintenance Circuit Decompositions}. (He also introduces a more restrictive class of solutions, ``Maintenance Euler Tours,'' where all the $W_i$ are required to be identical, but even when strongly periodic solutions exist, such ``Maintenance Euler Tours'' do not necessarily exist; see~\cite[Figure 5]{Talluri2015}.) Talluri~\cite{Talluri2015} has shown how to decide in polynomial time the existence of a Maintenance Circuit Decomposition when $\gamma \leq 4$. However, his result does not show that the periodic aircraft routing problem is polynomial in that case since there is no a priori reason that a periodic solution exists when the problem is feasible. As noted earlier, no existing result in the literature establishes a direct link between the existence of feasible solutions to the existence of periodic solutions.

We provide the first result of this type.

\begin{theorem}\label{thm:k4}
Consider an instance of the periodic aircraft routing problem with $\gamma \leq 4$. If the problem is feasible, then there exists an absolutely periodic solution.
\end{theorem}

In particular, this implies that the periodic aircraft routing problem is polynomial when $\gamma \leq 4$. What happens when $\gamma \geq 5$ is an open question.

\subsection{Finite-horizon}
In this section, we introduce the \emph{finite-horizon aircraft routing problem}, for which we will establish our hardness result. This version involves a finite collection of flight legs. Each flight leg is defined by a specific day within the planning horizon, a departure and arrival airport, and a departure and arrival time, and it must be performed exactly once. This version differs from the periodic one, where the same flight schedule repeats every day. The input of the problem is formed by
\begin{itemize}
    \item a set $\airports$ of \emph{airports}, a subset $B$ of which being \emph{bases},
    \item for each base in $B$, a time interval over which maintenance can be performed, and the duration of a typical maintenance at this base,
    \item a collection $L$ of \emph{flight legs}, each of them being characterized by departure and arrival airports, and by departure and arrival times (included in the horizon),
    \item a number $n_{\alpha}$ of \emph{airplanes} in each airport $\alpha$ at the beginning of the horizon, and
    \item a maximal number $\gamma$ of days with no maintenance.
\end{itemize}
The output is a partition of $L$ into $n$ \emph{routes}, a route being a sequence of flight legs that can be ``realized'': two flight legs can be chained if the first arrives in an airport from which the second departs, and the arrival time of the first is before the departure time of the second. The routes in the partition must moreover satisfy the maintenance constraint: an airplane must undergo a maintenance at most every $\gamma$ days.

To ease the definition of the problem, we assume that all times are given in a same time zone, typically the local time of the airline.

Note that without the maintenance constraint, the problem is polynomial since it can be modeled as a flow problem. We show in Section~\ref{sec:AR} that adding this constraint makes the problem difficult. As mentioned above, the only hardness result about the aircraft routing problem available in the literature is about the existence of Maintenance Euler Tours, which is a variation of the periodic aircraft routing problem.

\begin{theorem}\label{thm:ARP}
The finite-horizon aircraft routing problem is \NP-complete, even for a fixed $\gamma \geq 4$.
\end{theorem}

The proof reduces a problem of partitioning the arcs of a directed graph into disjoint paths---which is known to be \NP-complete---to the finite-horizon aircraft routing problem. 

The case $\gamma \leq 3$ is not discussed and this could be an interesting theoretical question to address. However, for the ``quiet night version''---a special case of the finite-horizon aircraft routing problem introduced in Section~\ref{subsec:polynomial}---the case $\gamma=1$ is obviously polynomial. Furthermore, the quiet night version is polynomial when the number of airplanes is fixed; see Section~\ref{subsec:polynomial}

\section{Periodicity}\label{sec:periodic}

As mentioned in the introduction, the existence of periodic solutions is not the issue.

\begin{proposition}\label{prop:periodic}
Consider an instance of the periodic aircraft routing problem. If the problem is feasible, then there exists a periodic solution.
\end{proposition}

\begin{proof}
Suppose there exists a solution $W_1,\ldots,W_m$. Denote by $a_{ij}$ the $j$th arc visited by $W_i$. Just before $W_i$ visits its $j$th arc, count the number of arcs traversed by $W_i$ since its last visit to a vertex in $B$, and denote this number by $k_{ij}$. There are finitely many possible tuples $\bigl((a_{1j},k_{1j}),\ldots,(a_{mj},k_{mj})\bigl)$. Thus, there exist $j_1 < j_2$ leading to the same tuple. A periodic solution $W'_1,\ldots,W'_m$ is obtained by defining $W'_i$ as the repeated copies of the sequence $a_{ij_1},a_{i(j_1+1)},\ldots,a_{ij_2-1}$.
\end{proof}

The proof only provides a rough bound of $m!\gamma^m$ on the period, leaving open the existence of periodic solutions with a polynomial description for all $\gamma$. (Note that Theorem~\ref{thm:k4} ensures the existence of periodic solutions with a polynomial description for $\gamma \leq 4$.)

\begin{figure}[h]
\begin{center}
			\psscalebox{1.0 1.0} 
			{
			\includegraphics{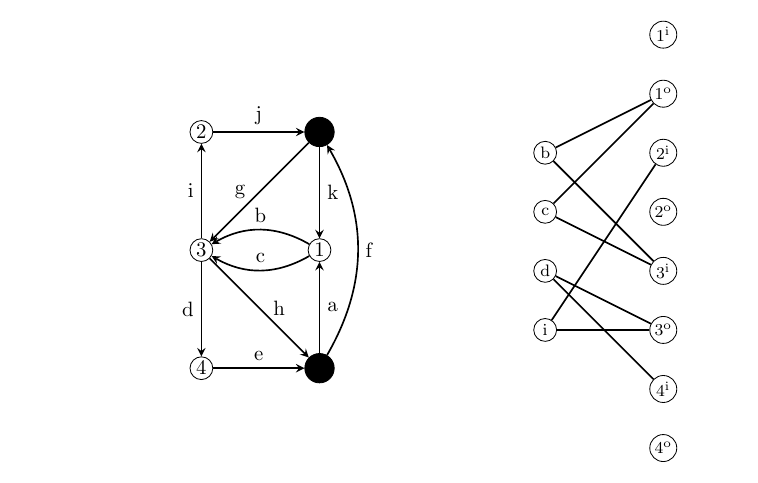}
			}
		\end{center}
\caption{On the left, an instance of the periodic aircraft routing problem; on the right, the corresponding graph $H$.\label{fig:H}}

\end{figure}

For each instance of the periodic aircraft routing problem, we build a bipartite graph $H$, which will be used at the very end of the proof of Theorem~\ref{thm:k4}, for the case $\gamma=4$. This graph provides a (polynomial) characterization of the existence of absolutely periodic solutions in that case; see Lemma~\ref{lem:bip} below. We describe now the construction. Write $H = (U, E)$, with bipartition $U = S \cup T$. The sets $S,T,E$ are now defined as follows. We set $S \coloneqq A \setminus \delta(B)$, meaning $S$ contains a vertex for each arc of $D$ that is not incident to a vertex in $B$. For each $v \in V \setminus B$, we introduce two vertices $v^{\text{i}}$ and $v^{\text{o}}$ in $T$, so that $T \coloneqq \{ v^{\text{i}}, v^{\text{o}} \colon v \in V \setminus B \}$. The edge set $E$ is defined such that an edge connects $a \in S$ to $v^{\text{i}} \in T$ if arc $a$ enters $v$ in $D$, and to $v^{\text{o}} \in T$ if arc $a$ leaves $v$ in $D$. The construction of $H$ is illustrated on Figure~\ref{fig:H}.

For each vertex $v \in V \setminus B$, let $b_v^-$ be the number of arcs from $B$ to $v$, let $b_v^+$ be the number of arcs from $v$ to $B$, and define $k_v \coloneqq \max(0, \deg^+(v) - b_v^- - b_v^+)$.

\begin{lemma}\label{lem:bip}
Consider an instance of the periodic aircraft routing problem and suppose that $\gamma =4$. There is an absolutely periodic solution if and only if there exists $F \subseteq E$ covering in $H$ each vertex of $S$ at most once and each vertex $v^{\text{\textup i}}$ and each vertex $v^{\text{\textup o}}$ at least $k_v$ times.
\end{lemma}

The special case $|B|=1$ has been established by Talluri~\cite[Proposition 3.1]{Talluri2015}, who somehow explains in the paragraph following its proof that the general case holds as well. To establish Theorem~\ref{thm:k4}, we need only the `if' direction, for which we provide a complete proof.

\begin{proof}[Proof of Lemma~\ref{lem:bip}, `if' direction]
Assume there is $F \subseteq E$ covering in $H$ each vertex of $S$ at most once and each vertex $v^{\text{\textup i}}$ and each vertex $v^{\text{\textup o}}$ at least $k_v$ times. Remove edges from $F$ (still calling this set $F$) so that each vertex $v^{\text{\textup i}}$ and each vertex $v^{\text{\textup o}}$ is covered exactly $k_v$ times by $F$.

Let $v$ be a vertex in $V \setminus B$. Denote by $A_{F,v}^+$ the arcs leaving $v$ and assigned to $v^{\text{o}}$ by $F$ and by $A_{F,v}^-$ the arcs entering $v$ and assigned to $v^{\text{i}}$ by $F$. Notice that $|A_{F,v}^+|=|A_{F,v}^-|= k_v$.

We describe now a process forming pairs of consecutive arcs. These pairs will not be necessarily disjoint, but they will cover all $A \setminus A[B]$ (where $A[B]$ denotes the arcs with both endpoints in $B$). We start with the arcs from the sets $A_{F,v}^+$ and $A_{F,v}^-$: for each vertex $v \in V \setminus B$, we greedily form disjoint pairs
with one element from $A_{F,v}^+$ and one element from $A_{F,v}^-$. Then, for each vertex $v \in V \setminus B$, we greedily form disjoint pairs with one element from $\delta^+(v) \setminus (A_{F,v}^+ \cup \delta^-(B))$ and one element from $B$ to $v$: this is possible because $\deg^+(v) - b_v^+ - |A_{F,v}^+| \leq \deg^+(v) - b_v^+ - k_v \leq b_v^-$ by the expression of $k_v$. Similarly, for each vertex $v \in V \setminus B$, we greedily form disjoint pairs with one element from $\delta^-(v) \setminus (A_{F,v}^- \cup \delta^+(B))$ and one element from $v$ to $B$: this is possible because $\deg^-(v) - b_v^- - |A_{F,v}^-| \leq \deg^+(v) - b_v^- - k_v\leq b_v^+$, again by the expression of $k_v$ and because $D$ is Eulerian. Finally, for each vertex $v \in V \setminus B$, we consider the arcs from $B$ to $v$ that have not been paired yet, and the arcs from $v$ to $B$ that have not been paired yet. Then, we greedily form disjoint pairs of one element from the first set and one element from the second set. This is possible since the graph $D$ is Eulerian. 

After this process, only the arcs in $A[B]$ have not been paired. The process above ensures the following for every arc in $A \setminus A[B]$:
\begin{itemize}
    \item If it belongs to $\delta^+(B)$, then it is paired with a unique arc leaving its head.
    \item If it belongs to  $\delta^-(B)$, then it is paired with a unique arc entering its tail.
    \item Otherwise, it is paired with a unique arc leaving its head, and also with a unique arc entering its tail.
\end{itemize}
The pairs form thus pairwise arc-disjoint directed trails from $B$ to $B$. Completing greedily these directed paths with arcs in $A[B]$ we get a partition of the arcs of $D$ into pairwise arc-disjoint closed directed trails. (Here, we use again that $D$ is Eulerian.)

To ease the discussion, we say that an arc of $D$ is at distance $q$ from $B$ if its tail is reached from $B$, on the trail the arcs belongs to, after exactly $q$ arcs. To finish the proof, we check that this distance is at most $3$ for every arc in $D$

The arcs leaving $B$ are at distance $0$. The arcs in $A_{F,v}^-$ for some $v$ are at distance $1$: indeed, each vertex of $S$ is covered at most once in $H$; hence, such an arc is not in $A_{F,u}^+$, where $u$ is its tail, and by the construction, it is preceded by an arc leaving $B$. The arcs in $A_{F,v}^+$ for some $v$ are at distance $2$ because they are preceded by an arc in $A_{F,v}^-$. Consider now an arc with no endpoints in $B$ and neither in some $A_{F,v}^-$, nor in some $A_{F,v}^+$. By the construction, it is preceded by an arc coming from $B$, and thus is at distance $1$. Finally, we are left with arcs entering $B$. By the construction, such arcs are preceded by an arc $a$ in a situation already dealt with, i.e., an arc $a$ at distance at most $2$, and we are done.
\end{proof}

We deal now with the proof of Theorem~\ref{thm:k4}. We deal successively with the different values of $\gamma$.

\begin{proof}[Proof of Theorem~\ref{thm:k4}, case $\gamma=1$]
     Suppose there exists a solution. Then $V=B$. Since $D$ is Eulerian, it can be decomposed into closed directed trails forming a partition of $A$, and whose vertices are all in $B$.
\end{proof}

\begin{proof}[Proof of Theorem~\ref{thm:k4}, case $\gamma=2$]
    Suppose there exists a solution. Then there is no arc between two vertices in $V\setminus B$. Since $D$ is Eulerian, it can be decomposed into closed directed trails forming a partition of $A$, and such that every visit to a vertex in $V\setminus B$ is followed by a visit to a vertex in $B$.
\end{proof}

\begin{proof}[Proof of Theorem~\ref{thm:k4}, case $\gamma=3$]
    Suppose there exists a solution. Denote by $A'$ the arcs whose endpoints are both in $V\setminus B$. Since there exists a solution, we have 
    \[
    b_v^- \geq |\delta^+(v) \cap A'| \quad \text{and} \quad  b_v^+ \geq |\delta^-(v) \cap A'|
    \]
    for every vertex $v \in V \setminus B$. Pair each arc $a \in A'$ arbitrarily with an arc from $B$ to its tail and with an arc from its head to $B$. By the inequalities above, each selected arc from $B$ or to $B$ is distinct. This process creates paths of length $3$. We remove all arcs from these paths, and replace each such path with an arc from the origin of the path (in $B$) to the end of the path (in $B$ as well). In the new graph, the in- and outdegrees of each vertex are equal because $D$ is Eulerian. This implies that the new graph can be decomposed into closed trails forming a partition of its arcs, which in turn decompose $D$ into closed trails forming a partition of $A$, such that each arc in $A'$ is preceded with an arc from $B$ and followed with an arc to $B$.
\end{proof}

\begin{proof}[Proof of Theorem~\ref{thm:k4}, case $\gamma=4$]
    Consider a solution $W_1,\ldots,W_m$  of the periodic aircraft routing problem. 
    Write each walk $W_i$ explicitly as $W_i = v_0^i,a_1^i,v_1^i,a_2^i,\ldots$, and define
    for each $i \in [m]$ and each positive integer $j$, the quantities
    \[
    q_j^i \coloneqq \left \{ \begin{array}{ll} j & \text{if $W_i$ restricted to the first $j$ arcs does not visit $B$,} \\[1.3ex] \min\{\ell \in [j] \colon v_{j-\ell}^i \in B\} & \text{otherwise,}  \end{array}\right.
    \]
    and
    \[
    f_{a,j}^q \coloneqq \frac 1 j \sum_{i=1}^m \sum_{\ell=1}^j \un(a_\ell^i = a) \un(q_\ell^i = q ) \, ,
    \]
    where $\un$ denote the indicator function.

    The quantity $q_j^i$ represents the number of arcs traversed since the last visit to $B$ when restricting $W_i$ to its first $j$ arcs. The definition of $q_j^i$ is adjusted to also cover the initial segment of $W_i$, when there has not been any visit to $B$ yet. (The exact definition of this adjustment does not really matter since the argument will rely on some asymptotic discussion.) The quantity $f_{a,j}^q$ is the frequency at which arc $a$ has been visited when the walks are restricted to their first $j$ arcs, distinguished according to the number $q$ of arcs since last visit to $B$.

    By definition, the following holds:
    \begin{enumerate}[label=(\roman*)]
        \item \label{positive} $f_{a,j}^1, f_{a,j}^2, f_{a,j}^3, f_{a,j}^4 \geq 0$ for every $a$ and $j$.\vspace{1mm}
        \item \label{un} $f_{a,j}^1 + f_{a,j}^2 + f_{a,j}^3 + f_{a,j}^4 = 1$ for every $a$ and $j$.
        \item \label{flow} $\ds{\sum_{a\in \delta^+(v)}f_{a,j}^3 \geq \sum_{a \in \delta^-(v)} f_{a,j}^2 - \frac m j \deg^-(v)}$ for every $v \notin B$ and $j$.
        \item \label{flow-bis} $\ds{\sum_{a \in \delta^-(v)} f_{a,j}^2 = \deg^-(v) - \sum_{a\in\delta^-(v)}f_{a,j}^1 - \sum_{a \in \delta^-(v)}f_{a,j}^3}$ for every $v \notin B$ and $j$.
        \item \label{outbasis} $\ds{\sum_{a\in\delta^-(v)}f_{a,j}^1 = \frac 1 j \deg^-(v) + \frac {j-1} j b_v^-}$ for every $v \notin B$ and $j$.
        \item \label{tobasis} $\ds{\sum_{a\in \delta^-(v)}f_{a,j}^3 \leq b_v^+}$ for every $v \notin B$ and $j$.
    \end{enumerate}
    Checking these relations is cumbersome but not difficult from the definition and using the fact that, for every $j$, each arc of the graph is the $j$th arc of exactly one walk.
    
    Choose arbitrarily an infinite increasing sequence $j_1 < j_2 < \cdots $ of positive integers making each $f_{a,j}^q$ converging. Such a sequence exists by compactness. Denote then $f_a^q = \lim_{\ell\to+\infty}f_{a,j_{\ell}}^q$. We have then for every vertex $v \notin B$
    \[
    \sum_{a\in \delta^+(v)}f_a^3 \geq \sum_{a \in \delta^-(v)} f_a^2 = \deg^-(v) - \sum_{a\in\delta^-(v)}f_a^1 - \sum_{a \in \delta^-(v)}f_a^3 \geq \deg^-(v) - b_v^- - b_v^+ \, .
    \]
    The first inequality, the equality, and the second inequality are respectively consequences of item~\ref{flow}, item~\ref{flow-bis}, and items~\ref{outbasis} and~\ref{tobasis} jointly.    Hence,
    \begin{equation}\label{eq:freq32}
    \sum_{a\in \delta^+(v)}f_a^3 \geq \sum_{a \in \delta^-(v)} f_a^2 \geq k_v \, .
    \end{equation}
    For each edge $e$ of $H$, define $x_e$ as follows. If $e$ is of the form $av^{\text{i}}$ with $a \in A$ and $v$ the head of $a$, set $x_e = f_a^2$. If $e$ is of the form $av^{\text{o}}$ with $a \in A$ and $v$ the tail of $a$, set $x_e = f_a^3$. The vector $x$ is a feasible solution of the following system of linear inequalities:
    \[
    \begin{array}{ll}
        \displaystyle{\sum_{e\in\delta(s)} x_e \leq 1} & \forall s \in S  \\[1.5ex]
         \displaystyle{\sum_{e\in\delta(t)} x_e \geq k_v} & \forall t \in T, \text{with $v$ such that } t = v^{\text{i}} \text{ or } t = v^{\text{o}}  \\
         x_e \geq 0 & \forall e \in E \, ,
    \end{array}
    \]
    where the constraints come respectively from items~\ref{un},~\eqref{eq:freq32}, and~\ref{positive}. The matrix of the constraints being totally unimodular, and the system being feasible, there exists a feasible integer solution. This integer solution describes a subset $F \subseteq E$ covering in $H$ each vertex of $S$ at most once and each vertex $v^{\text{\textup i}}$ and each vertex $v^{\text{\textup o}}$ at least $k_v$ times. By Lemma~\ref{lem:bip}, there is an absolutely periodic solution.
\end{proof}

\section{Complexity of the finite-horizon version}\label{sec:AR}

\subsection{An intermediate graph problem}

We introduce in this subsection a pure combinatorial problem on a graph, which will be useful to address the complexity of the finite-horizon aircraft routing problem. Indeed, as we will see, this combinatorial problem can be reduced in polynomial time to the finite-horizon aircraft routing problem, which will help show the hardness of the latter; we will also show that the finite-horizon aircraft routing problem can be reduced  in polynomial time to this combinatorial problem, which will help get some tractability result.

The notation used in this section is independent of the notation of Section~\ref{sec:periodic}.

Consider a directed acyclic graph $D=(V,A)$ whose vertices are of three types:
\begin{itemize}
    \item {\em sources}, denoted by $S$, which are vertices with indegree $0$ and outdegree $1$.
    \item {\em sinks}, denoted by $T$, which are vertices with indegree $1$ and outdegree $0$.
    \item other vertices, whose in- and outdegree are equal.
\end{itemize}
We are given $m$ subsets of vertices $X_1,\ldots,X_m$ such that $T \subset X_m \subset X_{m-1} \subset \cdots \subset X_1 \subset V$, such that $S \subseteq V \setminus X_1$, and such that the outneighborhood of $X_i \setminus X_{i+1}$ is included in $X_{i+1} \setminus X_{i+2}$ for all $i = 0,1,\ldots,m$ (where $X_0 \coloneqq V$, $X_{m+1} \coloneqq T$, and $X_{m+2} \coloneqq \varnothing$). Notice that the sets $X_i$ are then closed in the graph terminology (which means that no arc leaves $X_i$.) Let $N \coloneqq \bigcup_{i=1}^{m+1}\delta^-(X_i)$ and $M$ be a subset of $N$.

We define the following {\em constrained path partition problem}: decide whether there exists a partition of the arcs of $D$ into $S$-$T$ paths, such that each path traverses an arc of $M$ after at most $\gamma-1$ traversals of $N \setminus M$, and build such a partition when it exists.

In the reductions, roughly speaking, the indices of the $X_i$ will correspond to days, the arcs in $N$ are related to successions of flight legs occurring on consecutive days ($N$ for ``nights''), and those in $M$ are related to such successions occurring in a base with enough time to perform maintenance ($M$ for ``maintenance'').

\subsection{Hardness: proof of Theorem~\ref{thm:ARP}}

As expected to get the hardness result, we establish the following fact.

\begin{lemma}\label{lem:XXXtorouting}
    The constrained path partition problem can be reduced in polynomial time to the finite-horizon aircraft routing problem.
\end{lemma}

\begin{proof}
    We suppose that we are given an instance of the constrained path partition problem. We build a set $\airports$ by identifying each vertex in $V$ with a distinct airport. The bases $B$ correspond to the heads of the arcs in $M$.

    We attache to each vertex $v$ in $X_i \setminus X_{i+1}$ (setting $X_0 \coloneqq V$ and $X_{m+1} \coloneqq T$) a date of the form (day, time), where the day is the index $i$, and the time is between $6$am and $8$pm and such that it is increasing along the arcs induced by $X_i \setminus X_{i+1}$. (Such times can easily be built by breadth-first-search.) We define $L$ as follows. Each arc gives rise to a flight leg. Its departure and arrival airports are given by the tail and head vertices; the departure time is that of the tail vertex; the arrival time is an arbitrary time between the time of the tail vertex and the time of the head vertex (and in case the arrival airport is an element of $B$, we set this time before $8$pm). Finally, we set  $n_{\alpha} = 1$ for each airport identified with a vertex in $S$ and $n_{\alpha} = 0$ for all other airports.

    We set the time interval for maintenance to [$8$pm--$6$am] for each airport, and the duration of maintenance to the length of this interval (but any length smaller would be fine). The  maximum number of days with no maintenance is $\gamma$.

    It is then clear that any feasible solution of the finite-horizon aircraft routing problem for this input translates into a feasible solution of the constrained path partition problem, and conversely.
\end{proof}

\begin{proposition}\label{prop:equi-diff}
    The constrained path partition problem is \NP-complete when $\gamma \geq 4$.
\end{proposition}

The proof relies on the hardness of the {\em two-commodity arc-disjoint path problem in an acyclic digraph}, which is \NP-complete; see \cite{even1976complexity}, where acyclicity is implicit, and also \cite[Theorem 19.7]{korte2011combinatorial}, where acyclicity is explicit. The input of this problem is formed by an acyclic digraph $D'=(V',A')$, with two sources $s_1,s_2$ and two sinks $t_1,t_2$, and in two integer numbers $d_1,d_2$. It consists in deciding whether there exists a collection of pairwise arc-disjoint paths, $d_1$ of them from $s_1$ to $t_1$ and $d_2$ of them from $s_2$ to $t_2$.

\begin{proof}[Proof of Proposition~\ref{prop:equi-diff}]
    We proceed by reduction from the two-commodity arc-disjoint path problem in an acyclic digraph $D'$. Given an input of this problem, we construct a new graph $D$ from $D'$ by adding new vertices and new arcs, as follows.
    
    We introduce $d_1$ vertices with an arc to $s_1$ and $d_2$ vertices with an arc to $s_2$. The former are denoted by $K_1$ and the latter by $K_2$. In addition, for each vertex $v$ with $\deg^-(v) < \deg^+(v)$, we introduce $\deg^+(v)-\deg^-(v)$ vertices with an arc to $v$. We denote them $K'$. Similarly, we introduce $d_1$ vertices with an arc from $t_1$ and $d_2$ vertices with an arc from $t_2$. The former are denoted by $L_1$ and the latter by $L_2$. In addition, for each vertex $v$ with $\deg^+(v) < \deg^-(v)$, we introduce $\deg^-(v)-\deg^+(v)$ vertices with an arc from $v$. We denote them by $L'$.

    For each vertex $k$ in $K_1 \cup K_2 \cup K'$, we introduce two new vertices $u_k$, $v_k$, and two new arcs $(u_k,v_k)$, $(v_k,k)$, so that we have a path of length two from a source $u_k$ to the vertex $k$. Similarly, for each vertex $l$ in $L_1 \cup L_2 \cup L'$, we introduce two new vertices $u_l$, $v_l$, and two new arcs $(l,u_l)$, $(u_l,v_l)$, so that we have a path of length two from the vertex $l$ to a sink $v_l$. We define $S$ to be all sources $u_k$ and $T$ to be all sinks $v_l$. This describes completely the graph $D = (V,A)$.
    
    Then we define 
    \begin{itemize}
        \item $X_1\coloneqq V \setminus \{u_k \colon k \in K_1 \cup K_2 \cup K'\}$.
        \item $X_2 \coloneqq X_1 \setminus \{v_k \colon k \in K_1 \cup K_2 \cup K'\}$.
        \item $X_3 \coloneqq X_2 \setminus (V' \cup K_1 \cup K_2 \cup K')$.
        \item $X_4 \coloneqq X_3 \setminus (L_1 \cup L_2 \cup L')$.
    \end{itemize}

    Finally, we set 
    \[
        M \coloneqq \{(u_k,v_k) \colon k \in K_1 \} \cup \{(v_k,k)\colon k \in K'\} \cup \{(l,u_l) \colon l \in L_2\} \cup \{(u_l,v_l) \colon l \in L_1 \}\, .
    \]
    The graph $D$ together with the $X_i$ and $M$ form a valid instance of the constrained path partition problem. We set $\gamma \coloneqq 4$, which will prove the $\NP$-completeness for the case $\gamma = 4$. The other cases with $ \gamma > 4$ can be obtained via a direct adaptation of this proof (e.g., by adding vertices between $V'$ and $L_1 \cup L_2 \cup L'$.)

    We finish the proof by checking that there is a feasible solution to the two-commodity arc-disjoint path problem in $D'$ if and only if there is a feasible solution to the constrained path partition problem on $D$.

    Suppose that the two-commodity arc-disjoint path problem on $D'$ has a solution. To obtain a feasible solution to the constrained path partition problem, we extend each of the $d_1$ paths going from $s_1$ to $t_1$ by adding the arcs $(u_k,v_k)$, $(v_k,k)$, and $(k,s_1)$ with $k \in K_1$ at the beginning, and the arcs $(t_1,l)$, $(l,u_l)$, and $(u_l,v_l)$ with $l \in L_1$ at the end. We proceed similarly for the $d_2$ paths going from $s_2$ to $t_2$. Because of the degree condition, the arcs of $D$ not covered by these paths can be partitioned into arc-disjoint paths starting from the $u_{k'}$ to the $v_{l'}$, with $k' \in K'$ and $l' \in L'$. All these paths are $S$-$T$ paths, partitioning the arcs of $D$, and each such path traverses an arc of $M$ after at most $\gamma-1$ traversals of $N \setminus M$.

    Conversely, suppose that the constrained path problem admits a feasible solution. Observe that any path going through $K_2$ must also go through $L_2$ to satisfy the requirement on $M$. Since all arcs of $M$ incident to $L_2$ are covered by paths going through $K_2$, any path going through $K_1$ must end in $L_1$ to satisfy the requirement on $M$. Since $|K_1|=d_1$ and $|K_2|=d_2$, the result follows.
\end{proof}

\begin{proof}[Proof of Theorem~\ref{thm:ARP}]
    The finite-horizon aircraft routing problem being obviously in \NP, the result is a direct consequence of Lemma~\ref{lem:XXXtorouting} and Proposition~\ref{prop:equi-diff}.
\end{proof}

\subsection{A polynomial case}\label{subsec:polynomial}
The special case of interest in this subsection, which we refer to as the ``quiet night'' version, can be described as follows. All bases are located in the same time zone. At every such base, the typical duration of maintenance corresponds to the length of the maintenance interval. All arrivals at a base occur before the maintenance interval begins, and all departures take place after it ends. We assume that the planning horizon starts the day before the first departure of a flight, outside the maintenance interval of any base, and end the day after the last arrival of a flight, again outside the maintenance interval of any base. In particular, any airplane that overnights at a base has sufficient time to undergo maintenance. Finally, we assume that all flights take less than 24h. (Note that the special case of the quiet night version reflects the real situation of many airlines.)

In the following, we show how to reduce the ``quiet night'' version to the constrained path partition problem.

\begin{lemma}\label{lem:routingtoXXX}
    The quiet night version of the aircraft routing problem can be reduced in polynomial time to the constrained path partition problem, with $|S|$ equal to the number of airplanes.
\end{lemma}

\begin{proof}
Given an arbitrary instance of the aircraft routing problem in its quiet night version, we construct a corresponding instance of the constrained path partition problem as follows. We introduce a set $S$ and a set $T$, each of size equal to the number of airplanes. We build a set $V$ as follows. We put in $V$ the sets $S$ and $T$, and also the $(\alpha,\tau)$ for every airport $\alpha$ and every time $\tau$ that is either the beginning of a day for $\alpha$ or the departure of a flight leg from $\alpha$. We do not consider vertices $(\alpha,\tau)$ for $\tau$ a time during the first or last days of the horizon, which are days on which no flight occurs at all, by definition of the quiet night version.
We introduce now the arcs of $D$:

\begin{enumerate}[label=(\roman*)]  
    \item\label{source} For each $ s \in S $, we add an arc of the form $(s,(\alpha,\tau))$, where $ \alpha$ is an airport and where $\tau$ is the earliest possible time for that $\alpha$. In this construction, we ensure that each $ (\alpha, \tau) $ has an indegree of $ n_{\alpha} $ and each $ s \in S $ has outdegree $1$. This is feasible since $ |S| = \sum_{\alpha \in \airports}  n_{\alpha}$ and each airplane has a departure time from its initial airport.  
    \item\label{leg} We add an arc of the form $ \bigl((\alpha, \tau), (\alpha', \tau')\bigr) $ for every flight leg $\ell$ departing from $ \alpha $ at time $ \tau $, where $\alpha'$ is the arrival airport of $\ell$ and where $\tau'$ is the earliest possible time for $ \alpha' $ after the arrival of $ \ell $.  
    \item\label{ground} We add copies of an arc $\bigl((\alpha, \tau), (\alpha, \tau')\bigr)$, where $\tau'$ is the first time after $\tau$ at airport $\alpha$. The number of copies is equal to the number of airplanes that are at $\alpha$ immediately after time $\tau$, which is fully determined by the input.
    \item\label{end} We add an arc $ ((\alpha, \tau), t) $ whenever a flight leg $ \ell $ departing from $ \alpha $ at time $ \tau $ arrives at an airport where there is no further time after the arrival of $ \ell $.  
    \item\label{balance} We add arcs of the form $ ((\alpha, \tau), t) $, where $ t \in T $, and $ \tau $ is the latest time from $ \alpha $. The number of arcs added is chosen so that $(\alpha, \tau)$ has an outdegree equal to the number of airplanes remaining at $\alpha$ after time $\tau$, plus the number of departures from $\alpha$ at $\tau.$ \end{enumerate}

(The last two steps are done in such a way that each $t$ has indegree $1$.) We set $m$ to be the total number of days of the horizon. For $i \in [m]$, we define $X_i$ as $T$ together with the pairs $(\alpha,\tau)$ where $\tau$ occurs on day $i$ or after. We set $X_0 \coloneqq V$, $X_{m+1} \coloneqq T$, and $X_{m+2} \coloneqq \varnothing$. The set $M$ are formed by those arcs of $N=\bigcup_{i=1}^{m+1}\delta^-(X_i)$ corresponding to overnighting at a base. More precisely, an arc of $N$ is in $M$ if the following extra conditions hold:
\begin{itemize}
    \item for an arc of type~\ref{source}, $\alpha$ is in $B$.
    \item for an arc of type~\ref{leg}, $\alpha'$ is in $B$.
    \item for an arc of type~\ref{ground}, $\alpha$ is in $B$.
    \item for an arc of type~\ref{end}, $\ell$ arrives at a base.
    \item for an arc of type~\ref{balance}, $\alpha$ is in $B$. 
    \end{itemize}

This way, it is immediate to check that we get a valid instance of the constrained path problem. 

Clearly, a feasible solution to the aircraft routing problem leads to a partition of the arcs of $D$ into $S$-$T$ paths, such that each path traverses an arc of $M$ after at most $\gamma - 1$ traversals of $N \setminus M$, and conversely.
\end{proof}

We establish now the polynomiality of the constrained path partition problem when $S$ is of fixed size, implying with Lemma~\ref{lem:routingtoXXX} the polynomiality of the quiet night version of the aircraft routing problem when the number of airplanes is fixed. The proof is inspired by the celebrated proof of polynomiality of the directed subgraph homeomorphism problem for directed acyclic graphs and a ``fixed pattern'' by Hopcroft, Fortune, and Wyllie~\cite{fortune1980directed}, who rely their approach to a ``pebbling game.''

Consider the following game played on the graph $D$ given in input of the constrained path partition problem. Suppose that we start with one pebble on each vertex in $S$. At each time step, consider every vertex with a number of pebbles equal to its outdegree, and move these pebbles along the arcs leaving their current vertex, in such a way that no two pebbles move along the same arc. The game stops when each vertex has a number of pebbles different from its outdegree. (We will see that it necessarily means that all pebbles are located on vertices in $T$.)

At this stage, the game is not particularly interesting, as the only choice to be made for pebbles moving from a vertex is the arcs along which they move (with the constraint of having exactly one pebble for each arc), and with indistinguishable pebble, all realizations of the game are identical. We will soon introduce an additional constraint, but first, we establish a few properties of the game that will be useful later. The set of arcs along which pebbles move at time step $i$ and the set of vertices occupied by pebbles at the end of time step $i$ are unambiguously determined. Let $A_i$ and $V_i$ denote these sets of arcs and vertices, respectively. Moreover, we set $V_0 \coloneqq S$.

We write $u \preceq v$ for two vertices $u, v \in V$ if there is an $u$-$v$ path in $D$. This determines a partial order since $D$ is acyclic. Without loss of generality, we assume that $D$ has no isolated vertex.

    \begin{lemma}\label{lem:min-deg}
        At the end of time step $i$, a vertex not in $T$ has a number of pebbles equal to its outdegree if and only if it is minimal for $\preceq$ among the vertices in $V_i$.
    \end{lemma}

    \begin{proof}
        We proceed by induction on $i$. This is obviously true for $i=0$. Suppose that it is true for some $i \geq 0$. 
        
        Consider a vertex $v \in V_{i+1} \setminus T$ minimal for $\preceq$, at the end of time step $i+1$. For every inneighbor $u$ of $v$, look at any $s$-$u$ path, with $s \in S$. By the minimality of $v$, there are no pebbles on any vertex of this path. Since there was a pebble on $s$ initially, each vertex on this path had at some point a number of pebbles equal to its outdegree. This means that at some point, there were pebbles on each inneighbor $u$ of $v$. By minimality of $v$, these pebbles have left these inneighbors at the beginning of time step $i+1$ or before. At the end of time step $i+1$, the number of pebbles that have arrived on $v$ is at least equal to the indegree of $v$. At this stage, these pebbles might have left $v$ but this is not the case as we explained now. Indeed, $v$ cannot be minimal for $V_{i'}$ for $i' \leq i$: if it were minimal for such a $i'$, by induction, there could not be any pebble on $v$ at the end of time step $i+1$. So, there are on $v$ a number of pebbles at least equal to the indegree of $v$. On the other hand, induction shows that pebbles leave a vertex at most once during the first $i$ time steps, and hence, there are on $v$ a number of pebbles at most equal to the indegree of $v$. Therefore, at the end of time step $i+1$, the number of pebbles on $v$ is exactly its outdegree.

        Conversely, consider a vertex $v$ not in $T$ with a number of pebbles equal to its outdegree, just before time step $i+2$. It belongs to $V_{i+1}$ by definition. Suppose for a contradiction that there is a vertex $w \prec v$ with at least one pebble $p$ on it. Look at an $s$-$v$ path going through $w$, with $s \in S$. Denote by $u$ the vertex just before $v$ on this path. There is a pebble $q$ on $v$ that comes from $u$: indeed, the number of pebbles on $v$ is equal to its indegree, and by induction, each inneighbor has pebbles leaving it at most once. When $q$ left $u$---say at time step $i'\leq i$---the pebble $p$ was on a vertex $w' \prec u$. This implies that the vertex $u$ was not minimal for $\preceq$ among the vertices of $V_{i'}$. Yet, the number of pebbles on $u$ was equal to its outdegree because $q$ left $u$; this contradicts the induction hypothesis.
    \end{proof}
    
    Lemma~\ref{lem:min-deg} shows that the game is well-defined, i.e., as long as there are pebbles not located on a vertex in $T$, some pebbles will move at the next time step.

    \begin{lemma}\label{lem:Ai}
        The $A_i$ form a partition of $A$.
    \end{lemma}

    \begin{proof}
        Each arc $a$ of $A$ is traversed at some point by a pebble: take any path from $S$ to the tail of $a$; for each vertex of this path, there is a time step at which there are pebbles leaving it, and in particular for the tail of $a$.

        Each arc $a$ of $A$ is traversed at most once by a pebble: when a pebble traverses the arc $a$, the tail of $a$ is in $V_i$ and according to Lemma~\ref{lem:min-deg}, there is no pebble located at any vertex smaller than this tail for $\preceq$.
    \end{proof}

\begin{figure}[h]
\begin{center}
			\psscalebox{1.0 1.0} 
			{
			\includegraphics{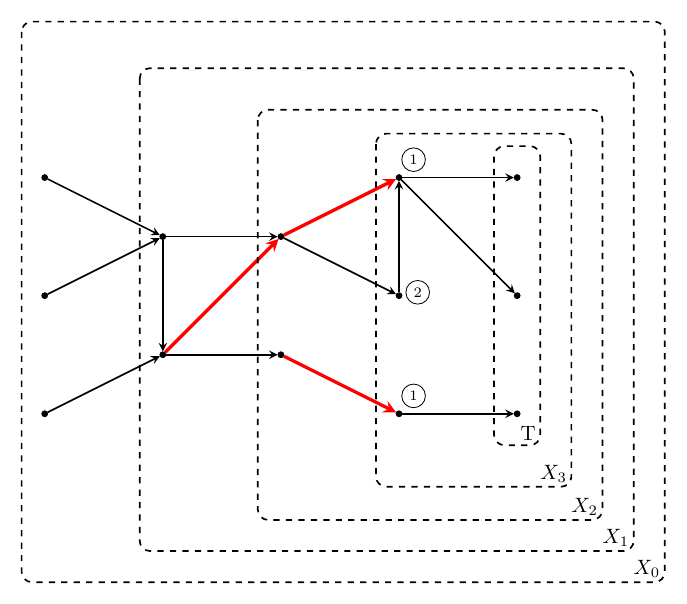}
			}
		\end{center}
\caption{An example of an instance of the constrained path partition problem. The arcs in $M$ are in red. A possible configuration of the ``pebble game'' is also illustrated, with three pebbles in $X_3$, and possible values of their counters after a few moves.\label{fig:game}}
\end{figure}

Now, we add the following feature: each pebble has a counter, initialized at $1$. Each time a pebble traverses an arc in $N\setminus M$, its counter is increased by one. Each time a pebble traverses an arc in $M$, its counter is set back to $1$. Each time a pebble traverses an arc not in $N$, its counter remains constant. To win the game, all pebbles must be located on the vertices in $T$, and all the moves must have been done in such a way that the counter of each pebble has never been larger than $\gamma$. Now, the choice of the arcs along which the pebbles move matters. (Actually, the pebbles can be interpreted as airplanes moving from airports to airports, while satisfying the maintenance constraint. So, going through the constrained path partition problem could be avoided to establish the polynomiality of the quiet night version of the aircraft routing problem. Nevertheless, each step gets easy statements.)

An illustration of a possible configuration of the game is given in Figure~\ref{fig:game}.

    \begin{lemma}\label{lem:equiv}
        The constrained path partition problem can be reduced in polynomial time to the problem of winning the game.
    \end{lemma}

    \begin{proof}
        Formally, we should start such a proof by building in polynomial time an instance of the game from an instance of the constrained path partition problem, but the instances are the same and so there is nothing to do. (In particular, the rest of the proof will actually establish that the two problems are polynomially equivalent.)
        
        Consider a solution of the constrained path partition problem. Let the game be at time step $i > 1$ and consider a pebble located at some vertex $v$ whose pebbles will leave it at this time step. Let $P$ be the path of the solution containing the arc along which the pebble arrived on $v$. Make the pebble traverse the next arc on $P$. Since each path of the solution traverses an arc of $M$ after at most $\gamma$ traversals of $N \setminus M$, following this decision for all time steps $i>1$ (for $i=1$, there is no decision to be taken), no pebble will have a counter exceeding $\gamma$ along the game.
        
        Consider a winning solution of the game. Each pebble determines a path $P$ from $S$ to $T$. These paths form a partition of $A$ because every arc is traversed exactly once by a pebble (Lemma~\ref{lem:Ai}). The counter of the pebble keeps track of the number of arcs of $N$ traversed by the path before traversing an arc of $M$.
    \end{proof}

We introduce now a graph $\calD=(\calV,\calA)$ encoding the possible moves of the pebbles on the graph. This graph is built from the graph $D=(V,A)$, and more precisely from the subsets $V_i$ and $A_i$ defined above. We denote by $r$ the largest index of the $V_i$. Notice that $V_r=T$ (by Lemma~\ref{lem:min-deg} and the remark after its proof).

The vertices $\calV$ are formed by all pairs $\bigl(i,(\boldx^v)_{v \in V_i}\bigl)$, where $i \in [r]$ and $\boldx^v \in\bbZ_+^{\gamma}$ with $\sum_{q=1}^{\gamma}x_q^v = \deg_{A_i}^-(v)$ for every $v \in V_i$.  We include also a vertex of the form $\bigl(0,(\bolds^v)_{v\in S}\bigl)$, with $\bolds^v$ defined for $v \in S$ by setting $s_1^v = 1$ and $s_q^v = 0$ for $q \geq 2$. The vertices of $\calD$ can be interpreted as the possible distributions of the pebbles, with the state of their counters, on the vertices of $D$ after $i$ steps.

    The arcs in $\calA$ are formed by all pairs $\left(\bigl(i,(\boldx^v)_{v \in V_i}\bigl),\bigl(i+1,(\boldy^v)_{v \in V_{i+1}}\bigl)\right)$ of vertices such that there is a map $\pi\colon A_{i+1} \to [\gamma]$  satisfying the following conditions:
    \begin{enumerate}[label=(\roman*)]
        \item\label{xvq} $x_q^v = \bigl|\{a \in \delta_{A_{i+1}}^+(v) \colon \pi(a) = q \}\bigl|$ for every $q \in [\gamma]$ and every $v \in V_i \setminus V_{i+1}$, \smallskip
        \item\label{yvq} $y_q^v = \bigl|\{a \in \delta_{A_{i+1}}^-(v)\setminus N \colon \pi(a) = q \}\bigl| + \bigl|\{a \in \delta_{A_{i+1}}^-(v)\cap (N\setminus M) \colon \pi(a) = q-1 \}\bigl|$ for $q \in \{2,\ldots,\gamma\}$ and every $v \in V_{i+1} \setminus V_i$, \smallskip
         \item\label{yv1} $y_1^v = \bigl|\{a \in \delta_{A_{i+1}}^-(v)\setminus N \colon \pi(a) = 1 \}\bigl| + | \delta_{A_{i+1}}^-(v)\cap M|$ for every $v \in V_{i+1} \setminus V_i$, and \smallskip
         \item\label{xy} $x_q^v = y_q^v$ for every $q \in [\gamma]$ and $v \in V_i \cap V_{i+1}$.
        \end{enumerate} Note that such a map $\pi$ determines uniquely the vectors $(\boldx^v)_{v \in V_i\setminus V_{i+1}}$ and $(\boldy^v)_{v \in V_{i+1}\setminus V_i}$. The arcs of $\calD$ can be interpreted as the possible changes of the distributions of the pebbles after one move. The value $\pi(a)$ can then be seen as the value of the counter of the pebble moving along the arc $a$, at its entrance.

\begin{lemma}\label{lem:win-game}
    Deciding whether there is a winning strategy and exhibiting such a strategy when it exists can be done in $O(|V|\gamma^{|S|})$.
\end{lemma}

\begin{proof}
     We prove that there is a way of winning if and only if there is a path in $\calD$ going from $\bigl(0,(\bolds^v)_{v\in S}\bigl)$ to a vertex of the form $\bigl(r,(\boldx^v)_{v\in V_r}\bigl)$. Once this is done, the conclusion follows easily: deciding the existence of such a path is linear in the number of arcs; the number of arcs is at most $r$ times the number of possible maps $\pi$, and hence at most $r\gamma^{|S|}$ since each $A_i$ has a cardinality upper bounded by $|S|$ (number of pebbles). The inequality $r \leq |V|$, which follows from the fact that no vertex is minimal in two distinct $V_i$, leads to expression given in the statement.

    Suppose there is a way of winning the game. For $v \in V_i$, define $\boldx^{iv}$ by setting $x_q^{iv}$ to be the number of pebbles located at $v$ with the counter at $q$. The vertices $\bigl(i,(\boldx^{iv})_{v \in V_i}\bigl)$ and $\bigl(i+1,(\boldx^{(i+1)v})_{v \in V_{i+1}}\bigl)$ are respectively the tail and the head of an arc in $\calA$. Indeed, as noted before the statement, setting $\pi(a)$ to the state of the counter of the pebble entering $a$ actually makes that the conditions from \ref{xvq} to \ref{xy} are satisfied.

    Conversely, suppose there is such a path in $\calD$. Denote by $\bigl(i,(\boldx^{iv})_{v \in V_i}\bigl)$ the vertices of this path, with $i \in \{0,1,\ldots,r\}$ (where $(\boldx^{0v})_{v \in V_0}= (\bolds^v)_{v\in S}$). There will be one time step for each integer $i$ in $[r]$. We prove by induction on $i$ that there is a way to realize a sequence of $i$ moves from the initial configuration of the pebbles to a configuration such that each vertex $v$ in $V_i$ has exactly $x_q^{iv}$ pebbles with the counter set at $q$, for each $q \in [\gamma]$ (and such that all pebbles are located on vertices in $V_i$). This is true by construction for $i=0$. Let us assume that it is true for some $i \leq r-1$. By induction, there are exactly $x_q^{iv}$ pebbles with the counter set at $q$, for each $q \in [\gamma]$ and each $v \in V_i$. 

    Consider the arc of $\calD$ from $\bigl(i,(\boldx^{iv})_{v \in V_i}\bigl)$ to $\bigl(i+1,(\boldx^{(i+1)v})_{v \in V_{i+1}}\bigl)$, with the corresponding map $\pi$. By definition of $\pi$, there is a way to move the pebbles along the arcs so that there is now exactly $x_q^{(i+1)v}$ pebbles with the counter set at $q$ on vertex $v$ of $V_{i+1}$, as we explain now.

    \begin{itemize}
        \item Let $u$ be a vertex in $V_i \setminus V_{i+1}$.  There are $x_q^{iu}$ pebbles on this vertex with the counter set at $q$, and there are $x_q^{iu}$ arcs $a$ leaving $v$ in $A_{i+1}$ such that $\pi(a)=q$ (condition \ref{xvq}). By definition, the pebbles on $u$ move at step $i+1$. For the $x_q^{iu}$ pebbles with the counter set at $q$, choose the arcs $a$ with $\pi(a)=q$. 
        \item Let $w$ be a vertex in $V_{i+1} \setminus V_i$. Let $q \in \{2,\ldots,\gamma\}$. The number of pebbles on $w$ after the move, with the counter set at $q$, is equal to the number of arcs $a$ in $\delta_{A_{i+1}}^-(w)$ with $\pi(a) = q$ if $a \notin N$ and with $\pi(a) = q-1$ if $a \in N \setminus M$ (condition~\ref{yvq}). The number of pebbles on $w$ after the move, with the counter set at $1$, is equal to the number of arcs $a$ in $\delta_{A_{i+1}}^-(w) \setminus N$ with $\pi(a) = 1$ plus the number of arcs in $\delta_{A_{i+1}}^-(w) \cap M$ (condition~\ref{yv1}).
        \item Let $w'$ be a vertex in $V_i \cap V_{i+1}$. The pebbles located on $w'$ are not moved, as we explain now. If $w' \in T$, this is clear. If $w' \notin T$, then Lemma~\ref{lem:min-deg} tells us that $w'$ will eventually have no pebbles, and $w'$ being in $V_{i+1}$ means that at least one pebble will arrive on $w'$ at the end of time step $i+1$ or later; Lemma~\ref{lem:min-deg} shows then that $w'$ is not minimal for $\preceq$. According to condition~\ref{xy} the number of pebbles on $w'$ with the counter set to $q$ at the end of time step $i+1$ is therefore equal to $x_q^{iw'}$, as desired.
    \end{itemize}
 Therefore, for all possible values of $q$, the number of pebbles on a vertex $w \in V_{i+1}$ with the counter set to $q$ after the $(i+1)$th move is equal to $x_q^{(i+1)w}$. 
\end{proof}

\begin{theorem}
    The quiet night version of the aircraft routing problem can be solved in polynomial time when the number of airplanes is fixed.
\end{theorem}

\begin{proof}
     Lemmas~\ref{lem:routingtoXXX},~\ref{lem:equiv}, and~\ref{lem:win-game} imply together the desired result.
\end{proof}

 \bibliographystyle{amsplain} 
 \bibliography{biblio}

\end{document}